\documentclass[journal,romanappendices]{IEEEtran}
\usepackage{algorithm,algorithmic,amsmath,amssymb,amsthm,bbm,cite,color,comment,enumitem,graphicx,microtype,subcaption,setspace,url}
\usepackage{mathtools,xparse}
\usepackage[USenglish]{babel}
\usepackage{hyperref}
\usepackage[utf8]{inputenc}
\usepackage[T1]{fontenc}
\usepackage{float}
\usepackage[autostyle]{csquotes}
\usepackage[labelsep=period,font=small]{caption}
\usepackage[nolist]{acronym}
\usepackage[normalem]{ulem}
\newtheorem{lemma}{Lemma}
\newtheorem{proposition}{Proposition}
\newtheorem{theorem}{Theorem}

\usepackage{tikz,pgfplots}
\usetikzlibrary{arrows, arrows.meta, backgrounds, calc, intersections, decorations.markings, decorations.pathreplacing, positioning, shapes,}
\pgfplotsset{compat=1.17}
\usepgfplotslibrary{fillbetween}
\newcommand{\Cross}{\mathbin{\tikz [x=1.4ex,y=1.4ex,line width=.2ex] \draw (0,0) -- (1,1) (0,1) -- (1,0);}}%
\usepackage{caption}
\usepackage{acronym}

\usepackage{titlesec}
\let\subparagraph\relax
\renewcommand{\a}{\mathbf{a}}
\renewcommand{\b}{\mathbf{b}}

\renewcommand{\r}{\mathbf{r}}
\newcommand{\s}{\mathbf{s}}

\renewcommand{\v}{\mathbf{v}}
\newcommand{\w}{\mathbf{w}}

\newcommand{\0}{\mathbf{0}}

\newcommand{\A}{\mathbf{A}}

\renewcommand{\H}{\mathbf{H}}
\newcommand{\I}{\mathbf{I}}

\newcommand{\R}{\mathbf{R}}

\newcommand{\W}{\mathbf{W}}

\newcommand{\thetab}{\boldsymbol{\theta}}

\newcommand{\sigmab}{\boldsymbol{\sigma}}

\newcommand{\Xib}{\mathbf{\Xi}}

\newcommand{\Sigmab}{\mathbf{\Sigma}}

\newcommand{\setD}{\mathcal{D}}
\newcommand{\setE}{\mathcal{E}}

\newcommand{\setJ}{\mathcal{J}}

\newcommand{\setL}{\mathcal{L}}

\newcommand{\setN}{\mathcal{N}}

\newcommand{\setT}{\mathcal{T}}
\newcommand{\setU}{\mathcal{U}}

\newcommand{\Compl}{\mbox{$\mathbb{C}$}}
\newcommand{\Real}{\mbox{$\mathbb{R}$}}

\newcommand{\diag}{\mathrm{diag}}

\newcommand{\diff}{\mathrm{d}}
\newcommand{\Exp}{\mathbb{E}}
\newcommand{\herm}{\mathrm{H}}

\newcommand{\tr}{\mathrm{tr}}
\newcommand{\tran}{\mathrm{T}}

\renewcommand{\vec}{\mathrm{vec}}
\newcommand{\card}{\mathrm{card}}

\definecolor{red}{rgb}{1,0,0}
\definecolor{red_magenta}{rgb}{1,0,0.5}
\definecolor{magenta}{rgb}{1,0,1}
\definecolor{blue_magenta}{rgb}{0.5,0,1}
\definecolor{blue}{rgb}{0,0,1}
\definecolor{blue_cyan}{rgb}{0,0.5,1}
\definecolor{cyan}{rgb}{0,1,1}
\definecolor{green_cyan}{rgb}{0,1,0.5}
\definecolor{green}{rgb}{0,1,0}
\definecolor{green_yellow}{rgb}{0.5,1,0}
\definecolor{yellow}{rgb}{1,1,0}
\definecolor{red_yellow}{rgb}{1,0.5,0}
\begin{acronym}
\acro{1D}{one-dimensional}
\acro{2D}{two-dimensional}
\acro{3D}{three-dimensional}
\acro{ACF}{autocorrelation function}
\acro{AWGN}{additive white Gaussian noise}
\acro{DoF}{degrees of freedom}
\acro{EM}{electromagnetic}
\acro{i.i.d.}{independent and identically distributed}
\acro{ITU}{International Telecommunication Union}
\acro{LoS}{line-of-sight}
\acro{MIMO}{multiple-input multiple-output}
\acro{NLoS}{non-line-of-sight}
\acro{OFDM}{orthogonal frequency-division multiplexing}
\acro{PSD}{power spectral density }
\acro{PSF}{power spectral factor}
\acro{SNR}{signal-to-noise ratio}
\acro{vMF}{von Mises-Fisher}
\acro{WDM}{wavenumber-division multiplexing}
\acro{PDF}{probability density function}
\acro{KL}{Kullback–Leibler}
\end{acronym}

\title{Analysis of Non-Separable Scattering in WDM-Based NLoS Holographic MIMO
\author{Ashutosh Prajapati, Prathapasinghe Dharmawansa, and Italo Atzeni,~\IEEEmembership{Senior Member, IEEE}}
\thanks{The authors are with the Centre for Wireless Communications, University of Oulu, Finland (e-mail: \{ashutosh.prajapati, prathapasinghe.kaluwadevage, italo.atzeni\}@oulu.fi).}
\thanks{This work was supported by the Research Council of Finland (336449 Profi6, 348396 HIGH-6G, and 369116 6G~Flagship).} \vspace{-0.5mm}}

\begin{document}

\maketitle

\begin{abstract}
Although scattering is inherently non-separable, holographic multiple-input multiple-output channels are commonly modeled using separable scattering across the source and receiver. In this paper, we investigate non-separable scattering in a non-line-of-sight holographic line model based on wavenumber-division multiplexing. We model its power spectral factor (PSF) using a bivariate von Mises-Fisher (vMF) distribution and introduce a non-separability parameter. We prove that the marginal densities flatten and their Kullback--Leibler divergence from the uniform distribution decreases monotonically with non-separability. Using majorization theory, we further establish that an upper bound and an approximation at low signal-to-noise ratio (SNR) of the ergodic spectral efficiency increase monotonically. Numerical results validate the analysis and demonstrate the performance gains of non-separable scattering over its separable counterpart.
\end{abstract}

\vspace{-1.5mm}

\begin{IEEEkeywords}
Holographic MIMO, non-separable scattering, wavenumber-division multiplexing.
\end{IEEEkeywords}                               

\vspace{-1.5mm}

\section{Introduction}
\label{sec:Intro}

The rapid growth in wireless data demand is driving communication systems toward higher frequencies, where dense deployments of physically small antennas become feasible \cite{Atz25}. This trend has motivated interest in holographic \ac{MIMO} as a promising technology for future high-capacity networks and intelligent applications. Unlike conventional massive \ac{MIMO}, which typically employs antenna elements spaced about half a wavelength apart to limit spatial correlation \cite{wei2026electromagnetic}, holographic \ac{MIMO} employs continuous apertures, such as surfaces or lines, or densely packed antenna arrays, yielding electrically large apertures. Such systems can be implemented using programmable metamaterials \cite{huang2020holographic}, offering flexible hardware architectures and multifunctional operation across the time, frequency, and spatial domains.

A continuous-aperture \ac{MIMO} model based on the \ac{WDM} concept was first investigated under purely \ac{LoS} propagation in \cite{sanguinetti2022wavenumber}, while \ac{NLoS} propagation was later considered in \cite{praj25journal}. Although being predominantly \ac{LoS}, millimeter-wave and sub-terahertz systems can experience significant small-scale fading, making \ac{NLoS} components non-negligible. These components arise from scattering that is generally non-separable across the source and receiver. Nevertheless, separable scattering, which yields the widely used Kronecker channel model, has been extensively adopted in classical and holographic \ac{MIMO} studies \cite{pizzo2020spatially, pizzo2022fourier, zhang2023capacity, kermoal2002stochastic}. Standardized channel models, such as 3GPP TR~38.901 \cite{3gpp38901v1700}, likewise employ separable scattering functions at the source and receiver. Non-separable scattering was first investigated in the context of classical \ac{MIMO} by Weichselberger \cite{weichselberger2006stochastic}, who provided the fundamental intuition for modeling transmit-receive coupling. More recently, non-separable scattering was considered for holographic \ac{MIMO} in \cite{zhang2025fundamental}, where the separable scattering functions are coupled through a Gaussian kernel. However, this approach does not provide a general framework for modeling the \ac{PSF} under non-separable scattering.

In this paper, we investigate non-separable scattering in the \ac{WDM}-based \ac{NLoS} holographic line model of \cite{praj25journal} and analytically establish its performance gains over separable scattering. Specifically, we introduce a non-separability parameter and model the \ac{PSF} characterizing the \ac{NLoS} channel using the sine model of the bivariate \ac{vMF} distribution \cite{demchuk2002miscellanea}. We show that the corresponding marginal densities become progressively flatter as the non-separability parameter increases. We quantify this behavior using the \ac{KL} divergence from the uniform distribution and prove that this divergence decreases monotonically with the non-separability parameter. Furthermore, we derive an upper bound and a low-\ac{SNR} approximation for the ergodic spectral efficiency and use majorization theory to prove that both increase monotonically with the non-separability parameter. Lastly, we validate the theoretical findings through numerical results. While the analysis establishes the monotonic increase of the ergodic spectral efficiency at low \ac{SNR}, the numerical results demonstrate that this monotonic behavior persists throughout the considered \ac{SNR} range.

\smallskip

\textbf{\textit{Notation.}} Bold lowercase and uppercase letters denote~vectors and matrices, respectively. $(\cdot)^{*}$, $(\cdot)^{\tran}$, and $(\cdot)^{\herm}$ denote conjugate, transpose, and Hermitian transpose, respectively. $[\A]_{a,b}$ is the $(a,b)$-th entry of $\A$. $\|\cdot\|_{2}$, and $\card(\cdot)$, $\vec(\cdot)$, and $\operatorname{mat}_{a \times b} (\cdot)$ represent Euclidean norm, cardinality, vectorization, and vector-to-matrix reshaping, respectively. $\det(\cdot)$, $(\cdot)^{\frac{1}{2}}$, $\diag(\cdot)$, and $\I_n$ denote determinant, element-wise square root of a matrix, diagonalization, and $n$-dimensional identity matrix, respectively. For two vectors $\a, \b \in \Real_{+}^n$, where $\Real_{+}$ is the set of nonnegative real numbers, $\a\succ\b$ denotes that $\a$ majorizes $\b$ \cite{marshall1979inequalities}. $\odot$, $\otimes$, and $\Cross$ are the Hadamard, Kronecker, and Cartesian products, respectively. $\mathbbm{1}_{\mathcal{D}}(\cdot)$ is the indicator function equal to $1$ for $x\in\mathcal{D}$ and $0$ otherwise. $|\cdot|$, $\left\lfloor\cdot\right\rfloor$, and $\delta[\cdot]$ are the absolute value, floor, and Kronecker delta functions, respectively. $j=\sqrt{-1}$ is the imaginary unit. $\setN_{\mathbb{C}}(\0,\A)$ is the circularly symmetric complex Gaussian distribution with zero mean and covariance matrix $\A$. $\mathbb{E}[\cdot]$ denotes expectation. $\phi(\cdot)$ and $\Phi(\cdot)$ are the standard Gaussian \ac{PDF} and cumulative distribution function, respectively.

\section{System Model}
\label{sec: sys}

\begin{figure}[t]
    \centering
    \begin{tikzpicture}[>=latex]

\small

\def\dh{0.68cm}
\def\dv{0.68cm}
\node[] (center) at (0,0) {};

\draw[->, thick, red] (-3,0) -- (-3,0.75) node[midway, sloped, above, xshift=2pt] {\footnotesize $x$-axis};
\draw[->, thick, red] (-3,0) -- (-2.25,0) node[midway, sloped, below, xshift=2pt] {\footnotesize $z$-axis};
\fill[red] (-3,0) circle (1pt);
\draw[red, thick] (-3,0) circle (3pt);
\node[rotate=45, anchor=east] at (-3.05,-0.05) {\footnotesize \textcolor{red}{$y$-axis}};

\node[left=\dh of center] (center_left) {};
\node[above=1.5*\dv of center_left] (top_left) {};
\node[below=1.5*\dv of center_left] (bottom_left) {};
\node[right=\dh of center] (center_right) {};
\node[above=1.5*\dv of center_right] (top_right) {};
\node[below=1.5*\dv of center_right] (bottom_right) {};

\draw[very thick] (top_left) -- (bottom_left) node[midway, right=1pt] {};
\draw[very thick] (top_right) -- (bottom_right) node[midway, left=1pt] {};

\node[draw, minimum width=1cm, minimum height=0.7cm, fill=gray!20, ellipse, rotate=60, transform shape] () at ($(center)+(0,0)$) {scatterer};

\node[draw, circle, left=\dh/4 of center_left, align=center, inner sep=0pt] (+_left) {+};
\draw[->] (+_left.east) -- (center_left.center);
\node[right=\dh/4 of center_right, align=center, inner sep=0pt] (cross_right) {};
\draw[-] (center_right.center) -- (cross_right.center);

\node[draw, rectangle, above left=1.25*\dv and \dh/3 of +_left.center, align=center, text width=1cm] (phi_1) {$\phi_1(s_x)$};
\node[draw, rectangle, below left=1.25*\dv and \dh/3 of +_left.center, align=center, text width=1cm] (phi_2) {$\phi_N(s_x)$};
\node[align=center] () at ($(+_left)+(-\dv+0.05cm,0cm)$) {$\vdots$};
\draw[->] (phi_1.east) -- ($(phi_1.east)+(\dh/3,0)$) -- (+_left.north);
\draw[->] (phi_2.east) -- ($(phi_2.east)+(\dh/3,0)$) -- (+_left.south);

\node[draw, rectangle, above right=1.25*\dv and \dh/3 of cross_right.center, align=center, text width=2cm] (psi_1) {$\int(\cdot)\psi_{1}^{*}(r_x)\diff r_x$};
\node[draw, rectangle, below right=1.25*\dv and \dh/3 of cross_right.center, align=center, text width=2cm] (psi_2) {$\int (\cdot)\psi_{N}^{*}(r_x)\diff r_x$};
\node[align=center] () at ($(cross_right)+(\dv+0.6cm,0cm)$) {$\vdots$};
\draw[->] (cross_right.center) -- ($(psi_1.west)+(-\dh/3,0)$) -- (psi_1.west);
\draw[->] (cross_right.center) -- ($(psi_2.west)+(-\dh/3,0)$) -- (psi_2.west);

\node[left=\dv/2 of phi_1, align=center, inner sep=0pt] (x_1) {$x_1$};
\node[left=\dv/2 of phi_2, align=center, inner sep=0pt] (x_2) {$x_N$};
\draw[->] (x_1.east) -- (phi_1.west);
\draw[->] (x_2.east) -- (phi_2.west);

\node[draw, circle, right=\dv/2 of psi_1, align=center, inner sep=0pt] (+_right_1) {+};
\node[draw, circle, right=\dv/2 of psi_2, align=center, inner sep=0pt] (+_right_2) {+};
\node[right=\dv/2 of +_right_1, align=center, inner sep=0pt] (y_1) {$y_1$};
\node[right=\dv/2 of +_right_2, align=center, inner sep=0pt] (y_2) {$y_N$};
\node[below=\dv/2 of +_right_1, align=center, inner sep=0pt] (z_1) {$z_1$};
\node[above=\dv/2 of +_right_2, align=center, inner sep=0pt] (z_2) {$z_N$};
\draw[->] (psi_1.east) -- (+_right_1.west);
\draw[->] (psi_2.east) -- (+_right_2.west);
\draw[->] (+_right_1.east) -- (y_1.west);
\draw[->] (+_right_2.east) -- (y_2.west);
\draw[->] (z_1.north) -- (+_right_1.south);
\draw[->] (z_2.south) -- (+_right_2.north);


\end{tikzpicture}
    \caption{Considered \ac{WDM}-based \ac{NLoS} holographic \ac{MIMO} system model.}
    \label{fig:WDM System Model}
\end{figure}
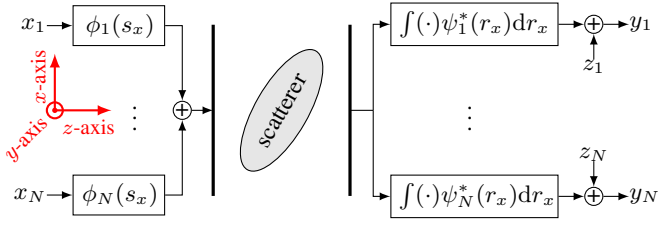

Consider a \ac{WDM}-based holographic \ac{MIMO} system over an \ac{NLoS} channel as depicted in Fig.~\ref{fig:WDM System Model}, where a line source spanning the linear region $\setL_{\textrm{s}} \subset \Real^{2}$ with length $L_{\textrm{s}}$ transmits data to a line receiver spanning the linear region $\setL_{\textrm{r}} \subset \Real^{2}$ with length $L_{\textrm{r}}$. The two lines are parallel and oriented along the $x$-axis, with their centers aligned along the $z$-axis and separated by a distance $d$. Let $\s = [s_{x}, s_{z}]^{\tran} \in \setL_{\textrm{s}}$ and $\r = [r_{x}, r_{z}]^{\tran} \in \setL_{\textrm{r}}$ denote arbitrary points within the source and receiver regions, respectively. For simplicity, we set $s_z=0$ and $r_z=d$; consequently, the source and receiver span the linear regions $\setL_{\textrm{s}}=\{(s_x, 0): |s_x|\leq \frac{L_{\textrm{s}}}{2}\} $ and $\setL_{\textrm{r}}= \{(r_x, d): |r_x|\leq \frac{L_{\textrm{r}}}{2}\}$, respectively, with corresponding gains $G_{\textrm{s}}=\frac{2L_{\textrm{s}}}{\lambda}$ and $G_{\textrm{r}}=\frac{2L_{\textrm{r}}}{\lambda}$ \cite{balanis2016antenna}. The wavenumber is defined as $\kappa = k = \frac{2\pi}{\lambda}$, where $\lambda$ is the wavelength. Let $\{x_m\}_{m=1}^{N}$ and $\{y_n\}_{n=1}^{N}$ denote the transmit and receive complex data symbols with $N \leq \mathrm{min}(n_{\textrm{s}}, n_{\textrm{r}})$, and let $z_n \in \setN_{\Compl}(0, \chi^2)$ be the additive white Gaussian noise term with variance $\chi^2$. Following \cite{praj25journal}, at any symbol time, the system model in Fig.~\ref{fig:WDM System Model} can be described~by 
\begin{align}
\label{eqn: input-output relationship}
y_n=\sqrt{G_{\textrm{s}}G_{\textrm{r}}}\sum_{m=1}^{N} H_{n, m} x_m+ z_n,\quad n=1,\ldots,N,
\end{align}
where 
\begin{align}
\label{eqn: Channel H_nm mod}
H_{n, m}=\int_{-\frac{L_{\textrm{r}}}{2}}^{\frac{L_{\textrm{r}}}{2}}\int_{-\frac{L_{\textrm{s}}}{2}}^{\frac{L_{\textrm{s}}}{2}}\psi_n^*(r_x) h(\r, \s) \phi_m(s_x) \diff s_x\diff r_x
\end{align}
is the \ac{WDM}-based \ac{NLoS} channel corresponding to the $m$-th transmit and $n$-th receive Fourier basis functions and $ h(\r, \s)$ is the \ac{EM}-based \ac{NLoS} channel impulse response \cite{praj25journal}. The transmit Fourier basis is $\big\{\phi_m(s_x)\big\}_{m=1}^{N}$, with $m$-th basis function $\phi_m(s_x)=\frac{1}{\sqrt{L_{\textrm{s}}}}e^{j\frac{2\pi}{L_{\textrm{s}}}m s_x},\ |s_x|\leq\frac{L_{\textrm{s}}}{2}$. 
Similarly, the receive Fourier basis is $\big\{\psi_n(r_x)\big\}_{n=1}^{N}$, with $n$-th basis function $ \frac{1}{\sqrt{L_{\textrm{r}}}}e^{j\frac{2\pi}{L_{\textrm{r}}}n r_x},\  |r_x|\leq \frac{L_{\textrm{r}}}{2}$.

Since $h(\r,\s)$ is bandlimited within $\setD^{2}$ \cite{praj25journal}, which specifies the channel support in the continuous wavenumber domain, sampling this support as in \cite{pizzo2022fourier} yields its discretized counterpart formed by the sets $\setE_{\textrm{s}} =\big\{p_x \in \mathbb{Z}:-\kappa \leq \frac{2\pi}{L_{\textrm{s}}}p_x\leq \kappa\big\}$ and $\setE_{\textrm{r}} =\big\{q_x \in \mathbb{Z}:-k \leq \frac{2\pi}{L_{\textrm{r}}}q_x\leq k \big\}$ at the source and receiver, respectively. Now, define $ n_{\textrm{s}} = \card(\setE_{\textrm{s}}) =\big\lfloor\frac{2 L_{\textrm{s}}}{\lambda}\big\rfloor$ and $n_{\textrm{r}} = \card(\setE_{\textrm{r}}) =\big\lfloor\frac{2 L_{\textrm{r}}}{\lambda}\big\rfloor$. Following the Fourier plane-wave series expansion of a spatially stationary complex Gaussian random process from \cite{pizzo2022fourier}, we have
\begin{align}
\label{eqn: Discrete NLoS}
h(\r,\s)\approx\sum_{q_x\in\setE_{\textrm{r}}}\sum_{p_x \in\setE_{\textrm{s}}} H_{\textrm{a}}({q_x, p_x})a_{\textrm{r}}(q_x, \r)a_{\textrm{s}}(p_x, \s),
\end{align}
with $a_{\textrm{s}}(p_x,\s)=e^{-j\left(\frac{2\pi}{L_{\textrm{s}}}p_x s_{x}+\gamma(p_x)s_{z}\right)}$, $a_{\textrm{r}}(q_x,\r)=e^{j\left(\frac{2\pi}{L_{\textrm{r}}}q_xr_{x}+\gamma(q_x)r_{z}\right)}$, $\gamma(p_x)=\sqrt{\kappa^{2}-\big(\frac{2\pi p_x}{L_{\textrm{s}}}\big)^2}$, and $\gamma(q_x)=\sqrt{k^{2}-\big(\frac{2\pi q_x}{L_{\textrm{r}}}\big)^2}$, and where $H_{\textrm{a}}(q_x, p_x)\sim \setN_{\mathbb{C}}\big(0,\sigma^2(q_x, p_x)\big)$ acts as a coupling coefficient characterizing the interaction between the transmit and receive plane waves.

We now substitute the Fourier series expansion of $h(\r, \s)$ in \eqref{eqn: Discrete NLoS} into \eqref{eqn: Channel H_nm mod}. Since $e^{j\frac{2\pi}{L_{\textrm{r}}}\gamma(q_x)r_z} H_{\textrm{a}}(q_x, p_x)e^{-j\frac{2\pi}{L_{\textrm{s}}}\gamma(p_x)s_z}$ is statistically equivalent to $H_{\textrm{a}}(q_x, p_x)$, we can express \eqref{eqn: Channel H_nm mod} as \cite{praj25journal}
\begin{align}
H_{n, m}=\sqrt{L_{\textrm{r}} L_{\textrm{s}}}H_{\textrm{a}}({n, m}), 
\label{eqn: Channel H_nm discrete without sinc}
\end{align}
with $n\in\setE_{\textrm{r}}$ and $m\in\setE_{\textrm{s}}$. Furthermore, \eqref{eqn: Channel H_nm discrete without sinc} can be written in matrix form as
\begin{align}  
\label{eqn:AngularChannelMatrixStochastic}
    \H = \sqrt{L_{\textrm{r}} L_{\textrm{s}}}\Xib^{\frac{1}{2}}\odot\W \in \Compl^{n_{\textrm{r}} \times n_{\textrm{s}}},
\end{align}
where $\W \in \mathbb{C}^{n_{\textrm{r}} \times n_{\textrm{s}}}$ is a random matrix with \ac{i.i.d.} $\setN_{\mathbb{C}}(0,1)$ entries and $\Xib \in \Real_+^{n_{\textrm{r}} \times n_{\textrm{s}}}$ acts as the coupling matrix containing the variances $\left\{\sigma^2(n, m) : n \in \setE_{\textrm{r}}, m \in \setE_{\textrm{s}} \right\}$, with \cite{praj25journal}
\begin{align}
    \label{eqn: sigma at receiver2 R}
    \sigma^2(n, m) = \frac{1}{(2\pi)^2}\int\int_{{\setT}_{\textrm{r}}(n)\Cross{\setT}_{\textrm{s}}(m)}\hspace{-4mm}A^2(\theta_{\textrm{r}}, \theta_{\textrm{s}})\diff\theta_{\textrm{r}}\diff\theta_{\textrm{s}},
\end{align}
where $\setT_{\textrm{r}}(n)=\big\{\arccos\big({\frac{\lambda}{L_{\textrm{r}}}(n+1)}\big), \arccos\big({\frac{\lambda}{L_{\textrm{r}}}n}\big)\big\}$ and $\setT_{\textrm{s}}(m)=\big\{\arccos\big({\frac{\lambda}{L_{\textrm{s}}}(m+1)}\big), \arccos\big({\frac{\lambda}{L_{\textrm{s}}}m}\big)\big\}$ denote the receive and transmit angular intervals, respectively. Here, the non-negative function $A^2(\theta_{\textrm{r}},\theta_{\textrm{s}})$ characterizes the scattering environment and is referred to as \ac{PSF}.

\section{Analysis of Non-Separable Scattering}
\label{sec: Analysis of Non-Separable Scattering}

In this section, we model the \ac{PSF} through its scaled version $\tilde{A}^2(\theta_{\textrm{r}}, \theta_{\textrm{s}} )=\frac{A^2(\theta_{\textrm{r}},\theta_{\textrm{s}})}{(2\pi)^2}$ and analyze the marginal densities, the \ac{KL} divergences from the uniform distribution, and their monotonicity. Based on this, we characterize the spectral efficiency and establish its monotonicity. Finally, we derive the spatial autocorrelation of the \ac{WDM}-based \ac{NLoS} channel.

\subsection{Modeling of the PSF}
\label{sec: Modeling of the PSF}

Motivated by the separable scattering model considered in \cite{praj25journal}, we model, for the first time, the joint \ac{PSF} that characterizes non-separable scattering using a bivariate \ac{vMF} distribution. Specifically, we express $\tilde{A}^2(\theta_{\textrm{r}}, \theta_{\textrm{s}})$ as\footnote{We consider only forward-traveling waves \cite{praj25journal}, i.e., $\theta_{\textrm{s}},\theta_{\textrm{r}}\in[0, \pi)$. Since \eqref{eqn: modeling PSF} is highly concentrated around the mean, extending the support from $[0, \pi)$ to $[-\pi,\pi)$ simplifies the subsequent analysis while having a negligible impact on the numerical results.}
\begin{align}
\label{eqn: modeling PSF}
    \tilde{A}^2(\theta_{\textrm{r}}, \theta_{\textrm{s}} )=p_{\Theta_{\textrm{r}}\Theta_{\textrm{s}}}(\theta_{\textrm{r}}, \theta_{\textrm{s}}), \quad \theta_{\textrm{s}},\theta_{\textrm{r}}\in[-\pi, \pi),
\end{align} 
where
\begin{align}
    \label{eqn:Bivariavte 2DdvMF}
    \hspace{-2mm} \begin{split}
          p_{\Theta_{\textrm{r}}\Theta_{\textrm{s}}}(\theta_{\textrm{r}}, \theta_{\textrm{s}})& = K_{\mu} \exp\big(\alpha_{\textrm{r}}\cos{(\theta_{\textrm{r}} - \bar{\theta}_{\textrm{r}}})+\alpha_{\textrm{s}}\cos{(\theta_{\textrm{s}} - \bar{\theta}_{\textrm{s}}})\\
          & \phantom{=} \ +\mu\sin{(\theta_{\textrm{r}} - \bar{\theta}_{\textrm{r}}})\sin{(\theta_{\textrm{s}} - \bar{\theta}_{\textrm{s}}})\big)
    \end{split}
\end{align}
is the \ac{PDF} of the bivariate \ac{vMF} distribution~\cite{demchuk2002miscellanea}, where $K_{\mu}=\big(4\pi^{2}\sum_{l=0}^{\infty}\binom{2l}{l} \big(\frac{\mu^{2}}{4\alpha_{\textrm{r}}\alpha_{\textrm{s}}}\big)^{l} I_{l}(\alpha_{\textrm{r}}) I_{l}(\alpha_{\textrm{s}})\big)^{-1}$ denotes the normalization constant and $I_l(\cdot)$ is the modified Bessel function of the first kind and order $l$. Here, $\alpha_{\textrm{s}}, \alpha_{\textrm{r}}\geq 0$ and $\bar{\theta}_{\textrm{s}}, \bar{\theta}_{\textrm{r}}\in[0,\pi)$ are the concentration parameters and mean scattering angles, respectively, at the source and receiver, whereas $\mu$ is the \textit{non-separability parameter} controlling the degree of non-separability between the source and receiver scattering. The \ac{PDF} in \eqref{eqn:Bivariavte 2DdvMF} is unimodal for $\mu^2<\mu_*^2=\alpha_{\textrm{s}}\alpha_{\textrm{r}}$ \cite{mardia2007bivariate}, corresponding to a single physical scattering cluster, while it becomes bimodal for $\mu^2\geq\mu_*^2$, representing two distinct clusters. Throughout this paper, we restrict our attention to the unimodal regime, i.e., $\mu\in[0,\mu_*)$. 

In the special case of $\mu=0$, the proposed model recovers the separable scattering model in \cite{praj25journal}. In this case, the \ac{PDF} in \eqref{eqn:Bivariavte 2DdvMF} factorizes into the product of its marginal densities as $p_{\Theta_{\textrm{r}}\Theta_{\textrm{s}}}(\theta_{\textrm{r}}, \theta_{\textrm{s}})=p_{\Theta_\textrm{r}}(\theta_{\textrm{r}})p_{\Theta_\textrm{s}}(\theta_{\textrm{s}})$, where $p_{\Theta_\textrm{s}}(\theta_{\textrm{s}})$ and $p_{\Theta_\textrm{r}}(\theta_{\textrm{r}})$ denote the \acp{PDF} of the univariate \ac{vMF} distributions modeling the separable scattering at the source and receiver, respectively. Consequently, substituting this factorization into \eqref{eqn: sigma at receiver2 R}--\eqref{eqn: modeling PSF} yields $\sigma^2(n, m)=\sigma^{2}_{\textrm{r}}(n)\sigma^{2}_{\textrm{s}}(m)$, where $\sigma^{2}_{\textrm{s}}(m)$ and $\sigma^{2}_{\textrm{r}}(n)$ are the variances at the source and receiver, respectively.

\subsection{Marginal Density, KL Divergence, and Monotonicity}
\label{sec:Marginal Density, KL Divergence and Monotonicity}

In this section, we focus on the receiver as the corresponding results for the source follow from analogous derivations and lead to the same conclusions. The marginal density of the scattering at the receiver is given by \cite{demchuk2002miscellanea}
\begin{align}
    \label{eqn: marginal at receiver}
    f_{\Theta_{\textrm{r}}}(\theta_{\textrm{r}})=2\pi K_{\mu} I_0\big(a(\theta_{\textrm{r}})\big)\exp{\big(\alpha_{\textrm{r}}\cos{(\theta_{\textrm{r}} \! - \! \bar{\theta}_{\textrm{r}}})\big)},
\end{align}
with $a(\theta_{\textrm{r}})=\sqrt{\alpha_{\textrm{s}}^2+\mu^2\sin^2{(\theta_{\textrm{r}} - \bar{\theta}_{\textrm{r}}})}$. We note that, in the special case of separable scattering, we have $f_{\Theta_{\textrm{r}}}(\theta_{\textrm{r}})=p_{\Theta_\textrm{r}}(\theta_{\textrm{r}})$. The \ac{KL} divergence between the distribution associated with the marginal density in \eqref{eqn: marginal at receiver} and the  uniform  distribution with \ac{PDF} $u_{\Theta_{\textrm{r}}}(\theta_{\textrm{r}})=\frac{1}{2\pi}$, $\theta_{\textrm{r}}\in[-\pi,\pi)$ is given by
\begin{align}
    \label{eqn: KL}
    \begin{split}
    D_{\textrm{KL}}(f_{\Theta_{\textrm{r}}}\|u_{\Theta_{\textrm{r}}})
    & =\int_{-\pi}^{\pi}
    f_{\Theta_{\textrm{r}}}(\theta_{\textrm{r}})
    \log_2\bigg(
    \frac{f_{\Theta_{\textrm{r}}}(\theta_{\textrm{r}})}
    {u_{\Theta_{\textrm{r}}}(\theta_{\textrm{r}})}
    \bigg)
    d\theta_{\textrm{r}} \\
    &=\log_2{(2\pi K_{\mu}^{\frac{1}{2}})^2}+\alpha_{\textrm{r}}\Exp_{f_{\Theta_{\textrm{r}}}}[\cos{\theta_{\textrm{r}}}]\\
    & \phantom{=} \ +\Exp_{f_{\Theta_{\textrm{r}}}}\big[\log_2{I_0\big(a(\theta_{\textrm{r}})\big)}\big]. 
    \end{split}   
\end{align}

In the large-concentration and unimodal regime, i.e., for $\alpha_{\textrm{s}},\alpha_{\textrm{r}}\gg 1$ and $\mu^2<\mu_*^2=\alpha_{\textrm{s}}\alpha_{\textrm{r}}$, \eqref{eqn:Bivariavte 2DdvMF} can be approximated by the \ac{PDF} of the bivariate Gaussian distribution with mean and covariance matrix given by \cite{demchuk2002miscellanea, savaux2020computation}
\begin{align}
    \label{eqn: mean and covariance}
    \thetab=
    \begin{bmatrix}
       \bar{\theta}_{\textrm{r}} \\
       \bar{\theta}_{\textrm{s}}
    \end{bmatrix},\quad
    \Sigmab=
    \begin{bmatrix}
       \frac{\alpha_{\textrm{s}}}{\alpha_{\textrm{s}}\alpha_{\textrm{r}}-\mu^2} & \frac{\mu}{\alpha_{\textrm{s}}\alpha_{\textrm{r}}-\mu^2}\\
       \frac{\mu}{\alpha_{\textrm{s}}\alpha_{\textrm{r}}-\mu^2}& \frac{\alpha_{\textrm{r}}}{\alpha_{\textrm{r}}\alpha_{\textrm{r}}-\mu^2}
    \end{bmatrix},
\end{align}
respectively. Consequently, the marginal density in \eqref{eqn: marginal at receiver} can also be approximated by the \ac{PDF} of the Gaussian distribution, i.e., ${\Theta_{\textrm{r}}}\overset{\textrm{approx.}}{\sim} \setN\big(\bar{\theta}_{\textrm{r}} ,\frac{\alpha_{\textrm{s}}}{\alpha_{\textrm{s}}\alpha_{\textrm{r}}-\mu^2}\big)$. In this setting, the marginal concentration parameter of $\Theta_{\textrm{r}}$ is given by $\alpha_{\textrm{r}}^{\textrm{eff}}=\alpha_{\textrm{r}}-\frac{\mu^2}{\alpha_{\textrm{s}}}$, which is strictly decreasing for $\mu\in[0,\mu_*)$, indicating that the marginal density becomes progressively flatter as $\mu$ increases. Subsequently, using the Gaussian approximation of the marginal density, \eqref{eqn: KL} can be approximated~as~$D_{\textrm{KL}}(f_{\Theta_{\textrm{r}}}\|u_{\Theta_{\textrm{r}}})\approx\frac{1}{2}\log_2{\alpha_{\textrm{r}}^{\textrm{eff}}}$, which is also strictly decreasing for $\mu\in[0,\mu_*)$, indicating that $f_{\Theta_{\textrm{r}}}$ approaches the uniform density as $\mu$ increases. In Section~\ref{sec: Numerical Results}, we show that the large-concentration approximation closely matches the exact marginal densities and their corresponding \ac{KL} divergences for practical values of the concentration parameters and mean scattering angles.

\subsection{Ergodic Spectral Efficiency}
\label{sec: Characterization of Spectral Efficiency}

Assuming equal power allocation across the transmit streams, the ergodic spectral efficiency (in bits/s/Hz) is given by
\begin{align}
    \label{eqn: Spectral Efficiency}  
    C_{\H}(\mu)
    = \Exp \bigg[
      \log_{2} \det\left(\I_{n_{\textrm{r}}}+\frac{\textrm{snr}}{n_{\textrm{s}}}
      \H\H^{\herm}\right) \bigg],
\end{align}
where $\H$ is given in \eqref{eqn:AngularChannelMatrixStochastic} and $\textrm{snr}=\frac{P}{\chi^2}$ denotes the transmit \ac{SNR}, with $P$ representing the total transmit power. To simplify the monotonicity proofs, we normalize $\H$ by omitting the constant term $\sqrt{L_{\textrm{r}}L_{\textrm{s}}}$ in Sections~\ref{sec: Characterization of Spectral Efficiency} and \ref{sec: Monotonicity of Spectral Efficiency}. Consequently, we have $\Exp[\tr(\H\H^\herm)]=\sum_{n,m=1}^{n_{\textrm{r}}, n_{\textrm{s}}}\big[\Xib\big]_{n,m}=1$, i.e., the total average channel power is constant, since $\Xib$ contains the variances obtained from the \ac{PSF}. Under this assumption, we next derive an upper bound and a low-\ac{SNR} approximation of the ergodic spectral efficiency.

\smallskip

\textbf{\textit{Upper bound.}} Using Jensen's inequality for concave functions, we have $C_{\H}(\mu) \leq C_{\H}^{\textrm{UB}}(\mu)$, with
\begin{align}
    \label{eqn: upper bound Spectral Efficiency}  
    C_{\H}^{\textrm{UB}}(\mu)=
      \log_{2} \det\bigg(\I_{n_{\textrm{r}}}+\frac{\textrm{snr}}{n_{\textrm{s}}}
      \Exp \big[\H\H^{\herm}\big]\bigg).
\end{align}
Define $\v(\mu) = \big[v_1(\mu),\ldots,v_{n_{\textrm{r}}}(\mu)\big]^\tran \in \Real_{+}^{n_{\textrm{r}}}$, where $v_n(\mu)=\sum_{m=1}^{n_\textrm{s}}[\Xib]_{n,m}$ denotes the sum over the $n$-th row of $\Xib$. Exploiting $\Exp \big[\H\H^{\herm}\big] = \diag \big(\v(\mu)\big)$ allows to simplify \eqref{eqn: upper bound Spectral Efficiency}~as
\begin{align}
    \label{eqn: upper bound Spectral Efficiency1}  
    C_{\H}^{\textrm{UB}}(\mu)=
      \sum_{n=1}^{n_{\textrm{r}}}\log_{2} \bigg(1+\frac{\textrm{snr}}{n_{\textrm{s}}}
      v_n(\mu)\bigg).
\end{align}

\smallskip

\textbf{\textit{Low-\ac{SNR} approximation.}}
Using the $2$nd-order Taylor series expansion of \eqref{eqn: Spectral Efficiency}, we have $C_{\H}(\mu) \approx C_{\H}^{\textrm{low}}(\mu)$, with
\begin{align}
    \label{eqn: low snr Spectral Efficiency}  
    C_{\H}^{\textrm{low}}(\mu) = \frac{\textrm{snr}}{n_{\textrm{s}}\ln{2}}\Exp\big[\tr(\H\H^\herm)\big]\!-\!\frac{\textrm{snr}^2}{2n_{\textrm{s}}^2\ln{2}}\Exp\big[\tr\big((\H\H^\herm)^2\big)\big].
\end{align}
Define $\w(\mu) = \big[w_1(\mu),\ldots,w_{n_{\textrm{s}}}(\mu)\big]^\tran \in \Real_{+}^{n_{\textrm{s}}}$, where $w_m(\mu)=\sum_{n=1}^{n_\textrm{r}}[\Xib]_{n,m}$ denotes the sum over the $m$-th column of $\Xib$. Exploiting $\Exp\big[\tr(\H\H^\herm)\big]=1$ and $\Exp\big[\tr\big((\H\H^\herm)^2\big)\big]=\big\|\v(\mu)\big\|_{2}^2+\big\|\w(\mu)\big\|_{2}^2$ allows us to simplify \eqref{eqn: low snr Spectral Efficiency} as 
\begin{align}
    \label{eqn: low snr Spectral Efficiency1}  
    C_{\H}^{\textrm{low}}(\mu) = \frac{\textrm{snr}}{n_{\textrm{s}}\ln{2}}-\frac{\textrm{snr}^2}{2 n_{\textrm{s}}^2\ln{2}}\big(\big\|\v(\mu)\big\|_{2}^2+\big\|\w(\mu)\big\|_{2}^2\big).
\end{align}

\smallskip

Next, we investigate the monotonicity of the above upper bound and low-\ac{SNR} approximation with respect to $\mu$.\footnote{A high-\ac{SNR} approximation of the ergodic spectral efficiency requires a closed-form expression for $\Exp\big[\log_{2}\det(\H\H^\herm)\big]$, making it challenging to establish monotonicity with respect to $\mu$. We leave this analysis for future work; however, our numerical results show that the ergodic spectral efficiency increases monotonically with $\mu$ throughout the considered \ac{SNR} range.}

\subsection{Monotonicity of the Spectral Efficiency}
\label{sec: Monotonicity of Spectral Efficiency}

\textbf{\textit{Monotonicity of $C_{\H}^{\textnormal{UB}}(\mu)$ in \eqref{eqn: upper bound Spectral Efficiency1}.}} Since $C_{\H}^{\textrm{UB}}(\mu)$ depends on $\mu$ through $\v(\mu)$, we utilize \eqref{eqn: sigma at receiver2 R}--\eqref{eqn: modeling PSF}, to express $v_n(\mu)$ as
\begin{align}
         v_n(\mu)=\sum_{m=1}^{n_\textrm{s}} \big[\Xib\big]_{n ,m}&=\int_{\tau_{\textrm{r}}(n)}\sum_{m=1}^{n_\textrm{s}}\int_{\tau_{\textrm{s}}(m)} p_{\Theta_{\textrm{r}}\Theta_{\textrm{s}}}(\theta_{\textrm{r}}, \theta_{\textrm{s}})\diff\theta_{\textrm{s}}\diff\theta_{\textrm{r}},\nonumber\\
         \begin{split}
          \label{eqn: marginal row sum}
         &=\int_{\tau_{\textrm{r}}(n)}f_{\Theta_{\textrm{r}}}(\theta_{\textrm{r}})\diff\theta_{\textrm{r}}.
    \end{split}
\end{align}
Thus, $\v(\mu)$ is determined solely by the marginal density at the receiver. Establishing the monotonicity of $C_{\H}^{\textrm{UB}}(\mu)$ requires determining the majorization ordering of $\v(\mu)$ as $\mu$ varies. This, in turn, requires a uniform partition, whereas $\setT_{\textrm{r}}$ is generally nonuniform. We therefore introduce the set of uniformly spaced bins $\setU_{\textrm{r}}(n) = \big\{\big({\frac{\lambda}{L_{\textrm{r}}}n}\big), \big({\frac{\lambda}{L_{\textrm{r}}}(n+1)}\big) \big\}$ and apply the change of variable $\theta_{\textrm{r}}=\arccos{u_{\textrm{r}}}$ to rewrite \eqref{eqn: marginal row sum} as
\begin{align}
    \label{eqn: marginal row sum1}      v_n(\mu)=\int_{\setU_{\textrm{r}}(n)}\frac{f_{\Theta_{\textrm{r}}}(\arccos{u_{\textrm{r}}})}{\sqrt{1-u_{\textrm{r}}^2}}\diff u_{\textrm{r}}.
\end{align}
Based on \eqref{eqn: marginal row sum1}, the following proposition establishes the majorization order of $\v(\mu)$; then, Lemma~\ref{prop:flatten} and Theorem~\ref{thm:monotonicity} characterize the monotonicity of $C_{\H}^{\textnormal{UB}}(\mu)$ with respect to $\mu$.

\begin{proposition}
\label{prop:flatten}
Assuming $\bar{\theta}_{\textnormal{r}}\notin\{0,\pi\}$, for $0\leq\mu<\mu'<\mu_*$, we have $\v(\mu)\succ\v(\mu')$, i.e., $\v(\mu)$ becomes progressively more uniform as $\mu$ increases. The same holds for $\w(\mu)$.
\end{proposition}
\begin{IEEEproof}
See Appendix~\ref{sec: Majorization Order}.
\end{IEEEproof}

\begin{lemma} 
\label{Lem: lemma1}
Since $\log_{2} \big(1+\frac{\textnormal{snr}}{n_{\textnormal{s}}}v_n(\mu)\big)$ is strictly concave in $v_n(\mu)$, $C_{\H}^{\textnormal{UB}}(\mu)$ is strictly Schur-concave in $\v(\mu)$. Hence, $\v(\mu)\succ\v(\mu')$ implies $C_{\H}^{\textnormal{UB}}(\mu')>C_{\H}^{\textnormal{UB}}(\mu)$ \cite{marshall1979inequalities}. 
\end{lemma}

\begin{theorem}
\label{thm:monotonicity}
Assuming $\bar{\theta}_{\textnormal{r}}\notin\{0,\pi\}$, $C_{\H}^{\textnormal{UB}}(\mu)$ is strictly increasing for $\mu\in[0,\mu_*)$.
\end{theorem}

\begin{IEEEproof}
Combining Proposition~\ref{prop:flatten} and Lemma~\ref{Lem: lemma1}, we have that $\mu<\mu'$ implies $\v(\mu)\succ\v(\mu')$, leading to $C_{\H}^{\textnormal{UB}}(\mu')>C_{\H}^{\textnormal{UB}}(\mu)$.
\end{IEEEproof}

\smallskip

\textbf{\textit{Monotonicity of $C_{\H}^{\textnormal{low}}(\mu)$ in \eqref{eqn: low snr Spectral Efficiency1}.}}
Since $C_{\H}^{\textnormal{low}}(\mu)$ depends on $\mu$ through $\v(\mu)$ and $\w(\mu)$, its monotonicity with respect to $\mu$ is characterized by the following lemma and theorem.

\begin{lemma}
\label{lem:schurconvex}
Since $\|\v(\mu)\|_{2}^2=\sum_{n=1}^{n_{\textnormal{r}}} v_n^2(\mu)$ is Schur-convex in $\v(\mu)$, $\v(\mu)\succ \v(\mu')$ implies
$\|\v(\mu)\|_{2}^2>\|\v(\mu')\|_{2}^2$\cite{marshall1979inequalities}. The same holds for $\w(\mu)$.
\end{lemma}
 
\begin{theorem}
\label{thm:lowsnr}
Assuming $\bar{\theta}_{\textnormal{r}}, \bar{\theta}_{\textnormal{s}}\notin\{0,\pi\}$, $C_{\H}^{\textnormal{low}}(\mu)$ is strictly increasing for $\mu\in[0,\mu_*)$.
\end{theorem}

\begin{IEEEproof}
Combining Proposition~\ref{prop:flatten} and Lemma~\ref{lem:schurconvex}, we have that $\mu<\mu'$ implies $\big\|\v(\mu)\big\|_{2}^{2} > \big\|\v(\mu')\big\|_{2}^{2}$ and $\big\|\w(\mu)\big\|_{2}^{2} > \big\|\w(\mu')\big\|_{2}^{2}$, leading to $C_{\H}^{\textnormal{low}}(\mu')>C_{\H}^{\textnormal{low}}(\mu)$.
\end{IEEEproof}

\subsection{Spatial Autocorrelation of \texorpdfstring{$H_{n, m}$}{}}
\label{sec: Autocorrelation of Channel H_{nm}}

Using \eqref{eqn: Channel H_nm discrete without sinc} and following \cite{praj25journal}, the spatial autocorrelation of the \ac{WDM}-based \ac{NLoS} channel is given by
\begin{align}
\label{eqn: Autocorrelation angular1}
    \hspace{-2mm} R_{mn, qp} \! = \! \mathbb{E}[H_{n, m} H^*_{p, q}] \! = \! L_{\textrm{r}} L_{\textrm{s}} \sigma^2(n, m)\delta[n \! - \! p]\delta[m \! - \! q],
\end{align}
which is represented in matrix form as
\begin{align}
\label{eqn: Autocorrelation Matrix different}
 \R= L_{\textrm{r}}L_{\textrm{s}}\diag(\Vec{(\Xib})) \in\mathbb{R}_{+}^{n_{\textrm{r}}n_{\textrm{s}}\times n_{\textrm{r}}n_{\textrm{s}}}.
\end{align}
Accordingly, the \ac{WDM}-based \ac{NLoS} channel matrix is given~by
\begin{align}
    \label{eqn: Complete channel matrix}
    \H\triangleq \operatorname{mat}_{n_{\textrm{r}}\times n_{\textrm{s}}}\big(\R^{\frac{1}{2}} \vec(\W)\big).
\end{align}

In the special case of $\mu=0$, which corresponds to separable scattering model discussed in Section~\ref {sec: Modeling of the PSF}, \eqref{eqn: Autocorrelation angular1} reduces to
\begin{align}
\label{eqn: Autocorrelation Matrix different separable}
 \R=  L_{\textrm{r}}L_{\textrm{s}}\diag({\sigmab}_{\textrm{r}})\otimes \diag({\sigmab}_{\textrm{s}}),
\end{align}
where $\sigmab_{\textrm{s}}\in\Real_{+}^{n_{\textrm{s}}}$ and $\sigmab_{\textrm{r}}\in\Real_{+}^{n_{\textrm{r}}}$ collect the variances  $\{\sigma_{\textrm{s}}^{2}(m):m\in\setE_{\textrm{s}}\}$ and $\{\sigma_{\textrm{r}}^{2}(n):n\in\setE_{\textrm{r}}\}$ at the source and receiver, respectively. Hence, \eqref{eqn: Autocorrelation Matrix different separable} recovers the Kronecker structure of the \ac{WDM}-based \ac{NLoS} channel described in \cite{praj25journal}.

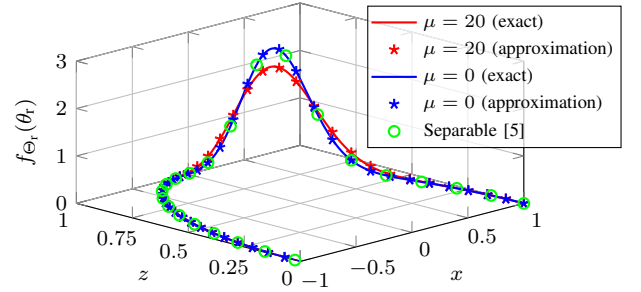
\begin{figure}[t]
\centering
\begin{tikzpicture}
\begin{axis}[
	width=7.5cm, height=5cm,
	xmin=-1, xmax=1,
	ymin=0, ymax=1,
	zmin=0, zmax=3,
	view={-45}{30},
	xlabel={$x$},
	ylabel={$z$},
	zlabel={$f_{\Theta_{\textrm{r}}}(\theta_{\textrm{r}})$},
	x label style={at={(axis description cs:0.85,0)}, anchor=north},
	y label style={at={(axis description cs:0.15,0)}, anchor=north},
	label style={font=\footnotesize},
	xtick={-1,-0.5,...,1},
	ytick={0,0.25,...,1},
	ztick={0,1,...,3},
	ticklabel style={font=\footnotesize},
	legend style={at={(0.65,0.99)}, anchor=north west, font=\scriptsize, inner sep=1pt, fill opacity=0.75, draw opacity=1, text opacity=1},
	legend cell align=left,
	title style={font=\scriptsize, yshift=-2mm},
	grid=major
]

\addplot3[thick, red] 
table[x=Var1, y=Var2, z=Marginal_Exact_r, col sep=comma]
{Results/files_txt/model_0_marginalR_sigma_mu_20.txt};
\addlegendentry{$\mu=20$ (exact)}

\addplot3[thick, red, only marks,  mark=star, mark options=solid, mark repeat=10] 
table[x=Var1, y=Var2, z=Marginal_Approx_r, col sep=comma]
{Results/files_txt/model_0_marginalR_sigma_mu_20.txt};
\addlegendentry{$\mu=20$ (approximation)}

\addplot3[thick, blue] 
table [x=Var1, y=Var2, z=Marginal_Exact_r, col sep=comma]
{Results/files_txt/model_0_marginalR_sigma_mu_0.txt};
\addlegendentry{$\mu=0$ (exact)}

\addplot3[thick, blue, only marks,  mark=star, mark options=solid, mark repeat=10] 
table[x=Var1, y=Var2, z=Marginal_Approx_r, col sep=comma]
{Results/files_txt/model_0_marginalR_sigma_mu_0.txt};
\addlegendentry{$\mu=0$ (approximation)}

\addplot3[thick, green, only marks, mark=o, mark options=solid, mark repeat=18] 
table[x=Var1, y=Var2, z=Separable_r, col sep=comma]
{Results/files_txt/model_0_marginalR_sigma_mu_0.txt};
\addlegendentry{\hspace{-0.8mm} Separable \cite{praj25journal}}
\end{axis}

\end{tikzpicture}
\caption{Marginal density at the receiver for different values of the non-separability parameter.}\label{fig:Fig1}
\end{figure}
\begin{figure}[t]
\centering
\begin{tikzpicture}
   \begin{axis}[
	width=7cm,
	height=4.75cm,
	xmin=0, xmax=35,
	ymin=0.4, ymax=1.6,
	xlabel={$\mu$},
	ylabel={$D_{\textrm{KL}}(f_{\Theta_{\textrm{r}}}\|u_{\Theta_{\textrm{r}}})$},
    xtick={0,5,...,35},
	ytick={0.4,0.8,...,1.6},
	xlabel near ticks,
	ylabel near ticks,
	x label style={font=\footnotesize},
	y label style={font=\footnotesize},
	ticklabel style={font=\footnotesize},
	legend style={at={(0.01,0.01)}, anchor=south west},
	legend style={font=\scriptsize, inner sep=1pt, fill opacity=0.75, draw opacity=1, text opacity=1},
	legend cell align=left,
	grid=both,
	title style={font=\scriptsize},
]  
\addplot[thick, red]
    table [x=lambda_values,y=KL_marg_exact_r, col sep=comma] {Results/files_txt/KL_Marginal_a_0.txt};
    \addlegendentry{Exact}
\addplot[thick, red, only marks, mark=star, mark options=solid, mark repeat=1] 
    table [x=lambda_values,y=KL_marg_asym_r, col sep=comma] {Results/files_txt/KL_Marginal_a_0.txt};
    \addlegendentry{Approximation}
    \end{axis}
\end{tikzpicture}
\caption{\ac{KL} divergence between the distribution associated with the marginal density in \eqref{eqn: marginal at receiver} and the uniform distribution versus non-separability parameter.}\label{fig:Fig2}
\end{figure}
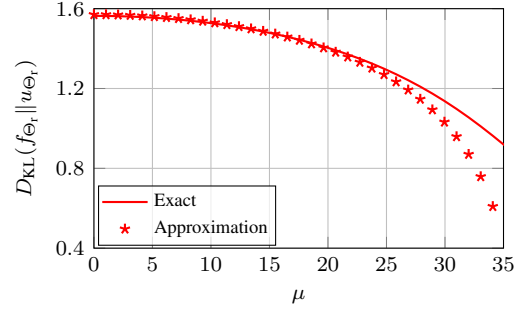
\begin{figure*}[t] 
    \centering
    \begin{minipage}[t]{0.32\textwidth}
        \centering
        \begin{tikzpicture}
   \begin{axis}[
	width=6.15cm,
	height=4.75cm,
	xmin=0, xmax=3000,
	ymin=-50, ymax=-10,
	xlabel={Eigenvalue index},
	ylabel={Normalized eigenvalue [dB]},
    xtick={1,500,1000,1500,2000,2500,3000},
    /pgf/number format/1000 sep={},
	ytick={-50,-40,...,-10,0},
	xlabel near ticks,
	ylabel near ticks,
	x label style={font=\footnotesize},
	y label style={font=\footnotesize},
	ticklabel style={font=\footnotesize},
	legend style={at={(0.99,0.99)}, anchor=north east},
	legend style={font=\scriptsize, inner sep=1pt, fill opacity=0.75, draw opacity=1, text opacity=1},
	legend cell align=left,
	grid=both,
	title style={font=\scriptsize},
]  
\addplot[thick, red]
table [x=index,y=Eigen_nonsep_mu, col sep=comma] {Results/files_txt/EigenvaluesR_a_0.txt};
\addlegendentry{$\mu=20$}
\addplot[thick, blue] 
table [x=index,y=Eigen_sep_mu, col sep=comma] {Results/files_txt/EigenvaluesR_a_0.txt};
\addlegendentry{$\mu=0$}
\addplot[thick, green, only marks, mark=o, mark options=solid, mark repeat=200] 
table [x=index,y=Eigen_sep_ref, col sep=comma] {Results/files_txt/EigenvaluesR_a_0.txt};
\addlegendentry{\hspace{-0.8mm} Separable \cite{praj25journal}}
\addplot[thick, magenta, densely dotted, mark=x, mark options=solid, mark repeat=30] 
table [x=index,y=Eigen_nonsep_ref, col sep=comma] {Results/files_txt/EigenvaluesR_a_0.txt};
\addlegendentry{\hspace{-0.8mm} Non-separable \cite{zhang2025fundamental}}
\end{axis}
\end{tikzpicture} \vspace{-2mm}
        \caption{Normalized eigenvalues of $\R$.}\label{fig:Fig3}
    \end{minipage}
    \hfill
    \begin{minipage}[t]{0.32\textwidth}
        \centering
        \begin{tikzpicture}
\begin{axis}[
     clip=false,
	width=6.15cm,
	height=4.75cm,
	xmin=-10, xmax=30,
	ymin=0, ymax=600,
	xlabel={\textrm{snr} [dB]},
	ylabel={$C_{\H}$ [bit/s/Hz]},
    xtick={-10,-5,...,30},
	xlabel near ticks,
	ylabel near ticks,
	x label style={font=\footnotesize},
	y label style={font=\footnotesize},
	ticklabel style={font=\footnotesize},
	legend style={at={(0.01,0.99)}, anchor=north west},
	legend style={font=\scriptsize, inner sep=1pt, fill opacity=0.75, draw opacity=1, text opacity=1},
	legend cell align=left,
	grid=both,
	title style={font=\scriptsize},
]  
\addplot[thick, red]
table [x=SNR_dB,y=Capacity_nonsep_mean, col sep=comma] {Results/files_txt/SE_SNR_mu_20_a_0.txt};
\addlegendentry{$\mu=20$}
\addplot[thick, blue] 
table [x=SNR_dB,y=Capacity_sep_mean, col sep=comma] {Results/files_txt/SE_SNR_mu_20_a_0.txt};
\addlegendentry{$\mu=0$}
\addplot[thick, green, only marks, mark=o, mark options=solid, mark repeat=2] 
table [x=SNR_dB,y=Capacity_sep_lit_mean, col sep=comma] {Results/files_txt/SE_SNR_mu_20_a_0.txt};
\addlegendentry{\hspace{-0.8mm} Separable \cite{praj25journal}}
\addplot[thick, magenta, densely dotted, mark=x, mark options=solid, mark repeat=5] 
table [x=SNR_dB,y=Capacity_iso_joint_lit_mean, col sep=comma] {Results/files_txt/SE_SNR_mu_20_a_0.txt};
\addlegendentry{\hspace{-0.8mm} Non-separable \cite{zhang2025fundamental}}

\addplot[thick, black] 
table [x=SNR_dB,y=Capacity_nonsep_ub_mean, col sep=comma] {Results/files_txt/SE_SNR_mu_20_a_0.txt};
\addlegendentry{Upper bound ($\mu=20$)}
\end{axis}
\end{tikzpicture} \vspace{-2mm}
        \caption{Ergodic spectral efficiency versus \ac{SNR}.}\label{fig:Fig4}
    \end{minipage}
    \hfill
    \begin{minipage}[t]{0.32\textwidth}
       \centering
        \begin{tikzpicture}
   \begin{axis}[
	width=6.15cm,
	height=4.75cm,
	xmin=0, xmax=35,
	ymin=100, ymax=500,
	xlabel={$\mu$},
	ylabel={$C_{\H}$ [bit/s/Hz]},
    xtick={0,5,...,35},
	ytick={100,200,...,500},
	xlabel near ticks,
	ylabel near ticks,
	x label style={font=\footnotesize},
	y label style={font=\footnotesize},
	ticklabel style={font=\footnotesize},
	legend style={at={(0.01,0.99)}, anchor=north west},
	legend style={font=\scriptsize, inner sep=1pt, fill opacity=0.75, draw opacity=1, text opacity=1},
	legend cell align=left,
	grid=both,
	title style={font=\scriptsize},
]  
\addplot[red,thick]
table[ x=mu,y=Capacity_nonsep,col sep=comma] {Results/files_txt/SE_mu_SNR_20_a_0.txt};
\addlegendentry{Monte Carlo}

\addplot[thick, black] 
table [x=mu,y=Capacity_noniso_joint_ub, col sep=comma] {Results/files_txt/SE_mu_SNR_20_a_0.txt};
\addlegendentry{Upper bound}

\addplot[red, thick]
table[ x=mu,y=Capacity_nonsep,col sep=comma] {Results/files_txt/SE_mu_SNR_10_a_0.txt};

\addplot[thick, black] 
table [x=mu,y=Capacity_noniso_joint_ub, col sep=comma] {Results/files_txt/SE_mu_SNR_10_a_0.txt};

\draw[black, thin](axis cs:25,340) ellipse [x radius=1, y radius=55];
\node[anchor=south] at (25,400) {\footnotesize{$\textrm{snr}=20$~dB}};

\draw[black, thin](axis cs:25,160) ellipse [x radius=1, y radius=35];
\node[anchor=south] at (25,200) {\footnotesize{$\textrm{snr}=10$~dB}};

\end{axis}
\end{tikzpicture} \vspace{-2mm}
        \caption{Ergodic spectral efficiency versus non-separability parameter, with $\textrm{snr} \in \{10, 20\}$~ dB.}\label{fig:Fig5}
    \end{minipage}
\end{figure*}

\section{Numerical Results and Discussion}
\label{sec: Numerical Results}

We consider holographic lines of lengths $L_{\textrm{s}} = L_{\textrm{r}} =64\lambda$, with $\lambda=0.01$~m (corresponding to a carrier frequency of $30$~GHz). The scattering environment is characterized by the mean scattering angles $\bar{\theta}_{\textrm{s}}=90^\circ$ and $\bar{\theta}_{\textrm{r}}=60^\circ$ and the normalized circular variances $\nu_{\textrm{s}}^2=0.03$ and $\nu_{\textrm{r}}^2=0.025$, as in \cite{pizzo2022fourier}. The corresponding concentration parameters are $\alpha_{\textrm{s}}=34$, obtained by solving the fixed-point equation $\nu_{\textrm{s}}^2=1-\big(\frac{I_1(\alpha_{\textrm{s}})}{I_0(\alpha_{\textrm{s}})}\big)^2$, and $\alpha_{\textrm{r}}=40$, obtained analogously from $\nu_{\textrm{r}}^2$. We consider only forward-traveling waves, i.e., $\theta_{\textrm{r}}\in[0, \pi)$.

Fig.~\ref{fig:Fig1} illustrates the marginal density at the receiver for $\mu=\{0,20\}$, obtained from \eqref{eqn: marginal at receiver}, and its large-concentration approximation. The exact and approximate densities agree closely. The case of $\mu=0$ corresponds to separable scattering and serves as a baseline hereafter. As established in Section~\ref{sec:Marginal Density, KL Divergence and Monotonicity}, the marginal density becomes flatter for $\mu=20$. Fig.~\ref{fig:Fig2} quantifies this behavior by showing the monotonic decrease of the \ac{KL} divergence from the uniform distribution. Fig.~\ref{fig:Fig3} plots the eigenvalues of the spatial autocorrelation matrix $\R$, normalized by its trace. The eigenvalues become more evenly distributed as $\mu$ increases from $0$ to $20$. In contrast, the non-separable scattering model based on \cite{zhang2025fundamental} yields the most uneven eigenvalue distribution, characterized by a rapid decay.

Fig.~\ref{fig:Fig4} plots the ergodic spectral efficiency in \eqref{eqn: Spectral Efficiency}, obtained via Monte Carlo simulations, versus the \ac{SNR}, showing improved performance as $\mu$ increases from $0$ to $20$. The upper bound in \eqref{eqn: upper bound Spectral Efficiency1} is fairly tight, especially at low \ac{SNR}, whereas the non-separable scattering model based on \cite{zhang2025fundamental} yields the lowest spectral efficiency due to its highly uneven eigenvalue distribution. Lastly, Fig.~\ref{fig:Fig5} depicts the ergodic spectral efficiency versus $\mu\in[0,35)$, with $\textrm{snr} \in \{10, 20\}$~dB,~where~both~the~upper bound and Monte Carlo results increase monotonically with $\mu$.

The above numerical results are consistent with the theoretical findings in Section~\ref{sec: Analysis of Non-Separable Scattering}. Future work may focus on deriving closed-form monotonicity results for the ergodic spectral efficiency at high \ac{SNR} and extending the analysis to holographic surfaces.




\appendices

\section{Proof of Proposition~\texorpdfstring{\ref{prop:flatten}}{}}
\label{sec: Majorization Order}

Since the integrand in \eqref{eqn: marginal row sum1} is the \ac{PDF} of the transformed variable $u_{\textrm{r}}=\cos(\theta_{\textrm{r}})$, applying the first-order delta method \cite{casella2024statistical} to $\zeta(\Theta_{\textrm{r}})=\cos(\Theta_{\textrm{r}})$, with ${\Theta_{\textrm{r}}}\overset{\textrm{approx.}}{\sim}\setN\big(\bar{\theta}_{\textrm{r}},\frac{\alpha_{\textrm{s}}}{\alpha_{\textrm{s}}\alpha_{\textrm{r}}-\mu^2}\big)$ (see Section~\ref{sec:Marginal Density, KL Divergence and Monotonicity}), gives $\zeta(\Theta_{\textrm{r}})\overset{\textrm{approx.}}{\sim}\setN\big(\zeta(\bar{\theta}_{\textrm{r}}),\big(\zeta'(\bar{\theta}_{\textrm{r}})\big)^2\frac{\alpha_{\textrm{s}}}{\alpha_{\textrm{s}}\alpha_{\textrm{r}}-\mu^2}\big)$. Therefore, the \ac{PDF} of $\zeta(\Theta_{\textrm{r}})$ is given by
\begin{align}
\label{eqn: pdf costheta}
    g_{\zeta(\Theta_{\textrm{r}})}(u_\textrm{r})
    =
    \frac{1}{\sqrt{2\pi}\rho}
    \exp\bigg(
    -\frac{(u_\textrm{r}-u_0)^2}{2\rho^2}
    \bigg),
\end{align}
with $u_0 \! = \! \cos{\bar{\theta}_{\textrm{r}}}$ and $\rho^2 \! = \! \sin^2\bar{\theta}_{\textrm{r}}\frac{\alpha_{\textrm{s}}}{\alpha_{\textrm{s}}\alpha_{\textrm{r}}-\mu^2}$. Thus, we have $\frac{\diff \rho^2}{\diff \mu} \! > \! 0$, i.e., $\rho^2$ is strictly increasing for $\mu\in[0,\mu_*)$ and $\bar{\theta}_{\textnormal{r}}\notin\{0,\pi\}$.

The \ac{PDF} in \eqref{eqn: pdf costheta} is symmetric about $u_0$ and strictly decreasing with respect to $|u_\textrm{r}-u_0|$. Accordingly, from \eqref{eqn: marginal row sum1}, a bin closer to $u_0$ has a larger probability mass than one farther from $u_0$. Let $\big\{v_{[i]}(\mu)\big\}_{i=1}^{n_{\textrm{r}}}$ be sorted in decreasing order and define $S_\ell(\mu)=\sum_{i=1}^{\ell}v_{[i]}(\mu)$. The $\ell$ bins that have the largest mass must form a contiguous block containing $u_0$. Hence, we have
\begin{align}
    \label{eq:proof}
    S_\ell(\mu)=\max_{j\in\setJ_\ell}
    \underbrace{\int_{a_j}^{b_j}g_{\zeta(\Theta_{\textrm{r}})}(u_{\textrm{r}})\diff u_{\textrm{r}}}_{= \ \Gamma (\rho)},
\end{align}
where $\setJ_\ell$ denotes the set of contiguous $\ell$-bin blocks $[a_j,b_j]$ containing $u_0$. The integral in \eqref{eq:proof} can be written as
\begin{align}
\label{eqn: Phi function int}
    \Gamma (\rho) = \Phi\bigg(\frac{b_j-u_0}{\rho}\bigg) + \Phi\bigg(\frac{u_0-a_j}{\rho}\bigg)-1,
\end{align}
and differentiating with respect to $\rho$ yields
\begin{align}
\begin{split}
    \hspace{-2mm} \frac{\diff}{\diff\rho} \Gamma (\rho) = -\frac{b_j \! - \! u_0}{\rho^2}\phi\bigg(\frac{b_j \! - \! u_0}{\rho}\bigg) \! - \! \frac{u_0-a_j}{\rho^2}\phi\bigg(\frac{u_0 \! - \! a_j}{\rho}\bigg).
\end{split}
\end{align}
Now, exploiting $b_j>u_0>a_j$ and $\phi(\cdot)>0$, we have $\frac{\diff}{\diff\rho}\int_{a_j}^{b_j}g_{\zeta(\Theta_{\textrm{r}})}(u_{\textrm{r}})\diff u_{\textrm{r}}<0$. In addition, $\frac{\diff\rho}{\diff\mu}>0$ implies $\frac{\diff}{\diff\mu}\int_{a_j}^{b_j}g_{\zeta(\Theta_{\textrm{r}})}(u_{\textrm{r}})\diff u_{\textrm{r}}<0$. Hence, every candidate $\ell$-bin block containing $u_0$ has decreasing probability mass as $\mu$ increases. Consequently, we have $S_\ell(\mu)>S_\ell(\mu')$ for $\mu<\mu'$ or, equivalently, $\sum_{i=1}^{\ell}v_{[i]}(\mu)>\sum_{i=1}^{\ell}v_{[i]}(\mu'),\quad\forall \ell=1,\ldots,n_{\textrm{r}}-1$, which implies $\v(\mu)\succ\v(\mu')$ \cite{marshall1979inequalities}. \hfill \IEEEQED

\addcontentsline{toc}{chapter}{References}
\bibliographystyle{IEEEtran}
\bibliography{refs_abbr, refs}

@book{balanis2016antenna,
  title={Antenna Theory: Analysis and Design},
  author={Balanis, Constantine A},
  year={2016},
  publisher={John Wiley \& Sons}
}

@article{pizzo2020spatially,
  title={Spatially-stationary model for holographic {MIMO} small-scale fading},
  author={Pizzo, Andrea and Marzetta, Thomas L and Sanguinetti, Luca},
  journal=J_IEEE_JSAC,
  volume={38},
  number={9},
  pages={1964--1979},
  year={2020},
  publisher={IEEE}
}

@article{pizzo2022fourier,
  title={Fourier plane-wave series expansion for holographic {MIMO} communications},
  author={Pizzo, Andrea and Sanguinetti, Luca and Marzetta, Thomas L},
  journal=J_IEEE_TWC,
  volume={21},
  number={9},
  pages={6890--6905},
  year={2022},
  publisher={IEEE}
}

@article{sanguinetti2022wavenumber,
  title={Wavenumber-division multiplexing in line-of-sight holographic {MIMO} communications},
  author={Sanguinetti, Luca and D’Amico, Antonio Alberto and Debbah, Merouane},
  journal=J_IEEE_TWC,
  volume={22},
  number={4},
  pages={2186--2201},
  year={2022},
  publisher={IEEE}
}

@article{zhang2025fundamental,
  author={Zhang, Xin and Song, Shenghui and Letaief, Khaled B.},
  journal=J_IEEE_TIT, 
  title={Fundamental Limits of Non-Centered Non-Separable Channels and Their Application in Holographic {MIMO} Communications}, 
  year={2025},
  volume={71},
  number={9},
  pages={6870--6894}
}

@article{huang2020holographic,
  title={Holographic {MIMO} surfaces for {6G} wireless networks: {Opportunities}, challenges, and trends},
  author={Huang, Chongwen and others},
  journal=J_IEEE_WC,
  volume={27},
  number={5},
  pages={118--125},
  year={2020},
  publisher={IEEE}
}

@article{wei2026electromagnetic,
  title={Electromagnetic information theory for holographic {MIMO} communications},
  author={Wei, Li and others},
  journal=J_IEEE_CST,
  year={2026},
  volume={28},
  pages={6211-6240},
  publisher={IEEE}
}

@article{praj25journal,
  title={{LoS+NLoS} Holographic {MIMO}: {Analysis} and Application of Wavenumber-Division Multiplexing},
  author= {Prajapati, A. and Dharmawansa, P. and Di Renzo, M. and Atzeni, I.},
  journal=J_IEEE_TWC,
  year={early access, 2026},
}

@article{Atz25,
	author = {I. {Atzeni} and others},
	title = {Sub-{THz} communications: {Perspective} and results from the {Hexa-X-II} project},
	journal = J_IEEE_OJCS,
	volume = {6},
	pages = {7495--7540},
	year = {2025}}

@article{zhang2023capacity,
  title={Capacity analysis of holographic {MIMO} channels with practical constraints},
  author={Zhang, Yuan and others},
  journal=L_IEEE_WCL,
  volume={12},
  number={6},
  pages={1101--1105},
  year={2023},
  publisher={IEEE}
}

@misc{3gpp38901v1700,
  author      = {{3GPP}},
  title       = {Study on channel model for frequencies from 0.5 to 100 {GHz}},
  institution = {3rd Generation Partnership Project},
  note        = {38.901 (v17.0.0)},
  year        = {2022}
}

@article{kermoal2002stochastic,
  title={A stochastic {MIMO} radio channel model with experimental validation},
  author={Kermoal, Jean-Philippe and others},
  journal=J_IEEE_JSAC,
  volume={20},
  number={6},
  pages={1211--1226},
  year={2002},
  publisher={IEEE}
}

@article{weichselberger2006stochastic,
  title={A stochastic {MIMO} channel model with joint correlation of both link ends},
  author={Weichselberger, Werner and others},
  journal=J_IEEE_TWC,
  volume={5},
  number={1},
  pages={90--100},
  year={2006},
  publisher={IEEE}
}

@article{demchuk2002miscellanea,
  title={Probabilistic Model for Two Dependent Circular Variables},
  author={Singh, Harshinder and Hnizdo, Vladimir and Demchuk, Eugene},
  journal={Biometrika},
  volume={89},
  no={3},
  pages={719--723},
  year={2002}
}

@article{mardia2007bivariate,
  title={Bivariate von {Mises} densities for angular data with applications to protein bioinformatics},
  author={Mardia, Kanti V and Taylor, Charles C and Subramaniam, Ganesh K},
  journal={Biometrics},
  volume={63},
  pages={505--512},
  year={2007}
}

@article{savaux2020computation,
  title={On the computation of integrals of bivariate {Gaussian} distribution},
  author={Savaux, Vincent and Le Magoarou, Luc},
  journal=C_ITG_WSA,
  year={2020},
  organization={IEEE}
}

@book{marshall1979inequalities,
  title={Inequalities: {Theory} of {Majorization} and its {Applications}},
  author={Marshall, Albert W and Olkin, Ingram and Arnold, Barry C},
  year={1979},
  publisher={Springer}
}

@book{casella2024statistical,
  title={Statistical {Inference}},
  author={Casella, George and Berger, Roger},
  year={2024},
  publisher={CRC press}
}

@STRING{J_IEEE_CST			= "IEEE Commun. Surveys and Tuts."}

@STRING{J_IEEE_TIT			= "IEEE Trans. Inf. Theory"}

@STRING{J_IEEE_JSAC			= "IEEE J. Sel. Areas Commun."}

@STRING{J_IEEE_OJCS			= "IEEE Open J. Commun. Soc."}

@STRING{J_IEEE_WC			= "IEEE Wireless Commun."}

@STRING{J_IEEE_TWC			= "IEEE Trans. Wireless Commun."}

@STRING{L_IEEE_WCL			= "IEEE Wireless Commun. Lett."}

@STRING{C_ITG_WSA			= "Proc. Int. ITG Workshop Smart Antennas (WSA)"}

@STRING{C_ITG_WSA			= "Proc. IEEE Symp. on Comput. and Commun. (ISCC)"}

\end{document}